\documentclass[review,authoryear]{elsarticle}
\usepackage{mathrsfs}

\usepackage{epsfig} 
\usepackage{latexsym}
\usepackage{pstricks}

\newtheorem{Def}{\em{Definition}}%[section]
\newtheorem{Theorem}{\em{Theorem}}%[section]
\newtheorem{Remark}{\em{Remark}}%[section]
\usepackage[hyphens]{url}
\usepackage{algorithm}
\usepackage{algorithmic}
\usepackage{amssymb, amsmath}
\usepackage{mathrsfs}
\usepackage{natbib}

\newcommand{\NP}{{\cal NP}}

\newenvironment{proof}{\noindent {\bf Proof. }}{$\blacksquare$ \vskip 1em}

\def \pbGen {$P|pmtn|f$}
\def \pbGenRelDates {$P|pmtn,r_j|f$} % ajout des r_j

\def \NP {{$\cal NP$}}
\def \P {{$\cal P$}}

\def \FS {$PFS-like$}

\begin{document}

\title{
New complexity results for parallel identical machine scheduling problems with preemption, release dates and regular criteria
} 

\author[emn]{D. Prot\corref{cor}}
\ead{damien.prot@mines-nantes.fr}
\author[emn]{O. Bellenguez-Morineau}
\ead{odile.morineau@mines-nantes.fr}
\author[emn]{C. Lahlou}
\ead{chams.lahlou@mines-nantes.fr}

\cortext[cor]{Corresponding author}

\address[emn]{ {\em{LUNAM Universit\'e, \'{E}cole des Mines de Nantes, IRCCyN UMR CNRS 6597 
(Institut de Recherche en Communication et en Cybern\'{e}tique de Nantes),
4 rue Alfred Kastler - La Chantrerie BP20722 ; 
F - 44307 NANTES Cedex 3 - FRANCE}}}

\date{Received: date / Accepted: date}

\begin{abstract} 
In this paper, we are interested in parallel identical machine scheduling problems with preemption and release dates in case of a regular criterion to be minimized. 
We show that solutions having a permutation flow shop structure are dominant if there exists 
an optimal solution with completion times scheduled in the same order as the release dates, or if there is no release date. 
We also prove that, for a subclass of these problems,
the completion times of all jobs can be ordered in an optimal solution. Using these two results, we provide
new results on polynomially solvable problems and hence refine the boundary between
$\cal P$ and \NP\ for these problems.
\end{abstract}


\maketitle

\section{Introduction}

%%%%%%%%%%%%%%%%%%%%%%%%%%%%%%%%%%%%%%%%%%
%%%%%%%%%%%%%%%% INTRO %%%%%%%%%%%%%%%%%%%
%%%%%%%%%%%%%%%%%%%%%%%%%%%%%%%%%%%%%%%%%%

\subsection{Definition of the problem}
The problem considered in this paper can be expressed as follows:
there are $n$ independent jobs $1, 2, \ldots, n$ and $m$ identical machines $M_1, M_2, \dots,M_m$.
Each job has a processing time $p_j$ and can be processed on any machine,
but only on one machine at a time. Preemption is allowed, meaning that 
a job can be interrupted and resumed on an other machine. There exists a release date $r_j$ for each job $j$, i.e. no job can start before its release date.
We are interested in minimizing a regular  (i.e. non-decreasing) function $f$. Usually, $f$ is either $\sum f_j$ or $\max f_j$, $f_j$ being any regular function of the completion time $C_j$ of job $j$.
Using the standard scheduling classification (\cite{GLLR89}), this problem is denoted \pbGenRelDates.
%We also study specific functions $f_i$'s to obtain 

%%%%%%%%%%%%%%%%%%%%%%%%%%%%%%%%%%%%%%%%%%
%%%%%%%%%%%% ETAT DE L'ART %%%%%%%%%%%%%%%
%%%%%%%%%%%%%%%%%%%%%%%%%%%%%%%%%%%%%%%%%%
\subsection{Related works}
Parallel identical machine scheduling problems are one of the most studied topic in scheduling theory.
For complexity results, the authors may refer to the websites maintained by~\cite{D11} and~\cite{BK11}.
A very recent survey on parallel machine problems with equal processing times, with or without preemption,
is produced by~\cite{KW11}. 
For different classical criteria, setting equal processing times makes a problem become polynomial-time solvable.
For example, \cite{BBCDKS07} prove that $P|pmtn, r_j, p_j=p|\sum C_j$ can be solved in polynomial-time whereas
$P|pmtn, r_j|\sum C_j$ is \NP-Hard (\cite{DLY90}).
For the total number of late jobs $\sum U_j$ criteria, the exact same behavior happens:
\cite{BBKT04} prove that $P|pmtn,p_j=p|\sum U_j$  is polynomial-time solvable,
and $P|pmtn|\sum U_j$  is \NP-Hard (see \cite{L83}).
For the total tardiness $\sum T_j$,
$P|pmtn, p_j=p|\sum T_j$ is solvable in polynomial-time (see \cite{BBKT04}) but
$P|pmtn|\sum T_j$ is \NP-Hard (see \cite{KW13}).
Adding weights on criteria make problems much more difficult since 
$P2|pmtn|\sum w_j T_j$ and $P|pmtn,p_j=p|\sum w_jU_j$ are \NP-Hard (see \cite{BCS74} and \cite{BK06}) 
Note that the complexity status of $P2|pmtn, p_j=p|\sum w_j T_j$ is still open.
Finally, by using linear programming techniques, results provided in \cite{LL78}, and also in \cite{BCSW76} and \cite{S81},
imply that $P|pmtn,r_j|L_{max}$ is solvable in polynomial-time.
For a survey on mathematical programming formulations in machine scheduling, the reader can refer to~\cite{BDW91}.

%%%%%%%%%%%%%%%%%%%%%%%%%%%%%%%%%%%%%%%%%%
%%%%%%%%%%%%%%%% PLAN %%%%%%%%%%%%%%%%%%%%
%%%%%%%%%%%%%%%%%%%%%%%%%%%%%%%%%%%%%%%%%%

%The main idea of this paper is to generalize two well-known approaches that have been 
%developed for special cases of our problem, and to derive more general results. 
\subsection{Contribution of the paper}

\cite{BBCDKS07} proved the existence of a dominant structure which allows to solve the problem $P|pmtn, r_j, p_j=p|\sum C_j$ in polynomial-time by linear programing.
In this paper, we extend this result to more general criteria and to some problems with non identical processing times, which implies new polynomial-time results, such as, to mention a few, for problems $P|pmtn, d_j=d|\sum T_j$, or $P|pmtn, r_j, p_j=p, d_j=d|\sum w_jU_j$.

More precisely, in section~\ref{sec:PFS}, we provide a dominant structure for all the problems of the form \pbGenRelDates\
for which there exists an optimal solution such that the completion times follow the same order as the release dates.
Note that our result implies that, for any problem of the type \pbGen\ (i.e. without release dates), 
the structure is dominant. 
We also prove that, if we are able to compute an order between optimal completion times of the different jobs, 
we can solve these problems in polynomial-time by linear programming. 
Section~\ref{sec:order} is dedicated to finding problems for which such an order exists. 
We discuss the implications of our results in term of complexity on classical criteria in section~\ref{sec:discussion}
and make some conclusions in section~\ref{sec:conclu}.

\section{A dominant structure}\label{sec:PFS}

We are looking at solutions having a \textit{Permutation Flow Shop-like} structure. In order to define this kind of schedules, we introduce some notations and concepts.

A \textit{piece} of a job is a part of the job that is scheduled without interruption. We say that a job $j$ is \textit{processed at time $t$} if there is a machine on which a piece of $j$ starts at time $t_1 \leq t$ and ends at time $t_2 > t$.
We denote by $C(j, t)$ (resp. $M(j,t)$) the completion time (resp. machine) of $j$ when it is processed at time $t$. For any time $t$, $J(t)$ denotes the set of jobs processed at time $t$.

A \textit{non-delay} schedule (called originally ``permissible left shift'', see \cite{GT60}) is such that, if a machine is idle during a time interval $[t,t+\epsilon[$ ($\epsilon > 0$), no piece of length $\epsilon' \leq \epsilon$ of a job processed at a time $t' > t$ can be processed at time $t$.

We also define a \textit{vertically ordered} schedule in the following manner: 
at any time $t$, if $J(t) = \{ j_1, j_2, \ldots, j_k \}$ and $j_1 < j_2 < \cdots j_k$,  we have $M(j_i,t)=M_i$, $i=1,2, \ldots, k$.

{Figure \ref{fig:NonDelayVerticalOrder} illustrates these properties: job $j_2$ verifies the non-delay property since none of its pieces can be processed earlier. On the contrary, the property is not verified by job $j_1$ because one piece can be scheduled at time $1$.
The schedule is vertically ordered during time interval $[3,4]$, but it is not during time interval $[2,3]$.
}

\begin{figure}[htbp]
\begin{center}
\includegraphics[scale=0.4]{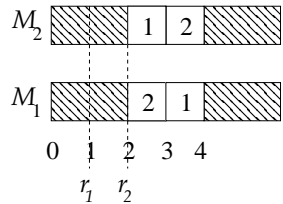}
\caption{Non-delay and vertical order properties}
\label{fig:NonDelayVerticalOrder}
\end{center}
\end{figure}

\begin{Remark}\label{Rem:NonDelay-VerticalOrder} When preemption is allowed, non-delay schedules are dominant for regular criteria since it is always possible to move a piece of a job to an earlier idle time interval without increasing the objective function. 
Moreover, vertically ordered schedules are also dominant since the completion times of the jobs remain the same if the pieces of jobs scheduled in the same time interval are reassigned to processors in order to respect the vertical order.
\end{Remark}

We assume without loss of generality that $r_1\leq r_2\leq \dots \leq r_n$ for the remainder of the paper.
Now, let us characterize the structure:

\begin{Def}
A schedule is said to be Permutation Flow Shop-like (\FS) if

\begin{enumerate}
\item it is vertically ordered,
\item no machine processes more than one piece of each job,
\item the scheduling order on the different machines is the same.
\end{enumerate}

%the scheduling order on the different machines is the same and each job is processed consecutively on machines $M_m, M_{m-1},\dots, M_{1}$ without preemption on any machine, some processing times possibly being null.
\end{Def}

An example of \FS\ schedule is given in Figure~\ref{fig:PFSGeneral}.

\begin{figure}[htbp]
\begin{center}
\includegraphics[scale=0.4]{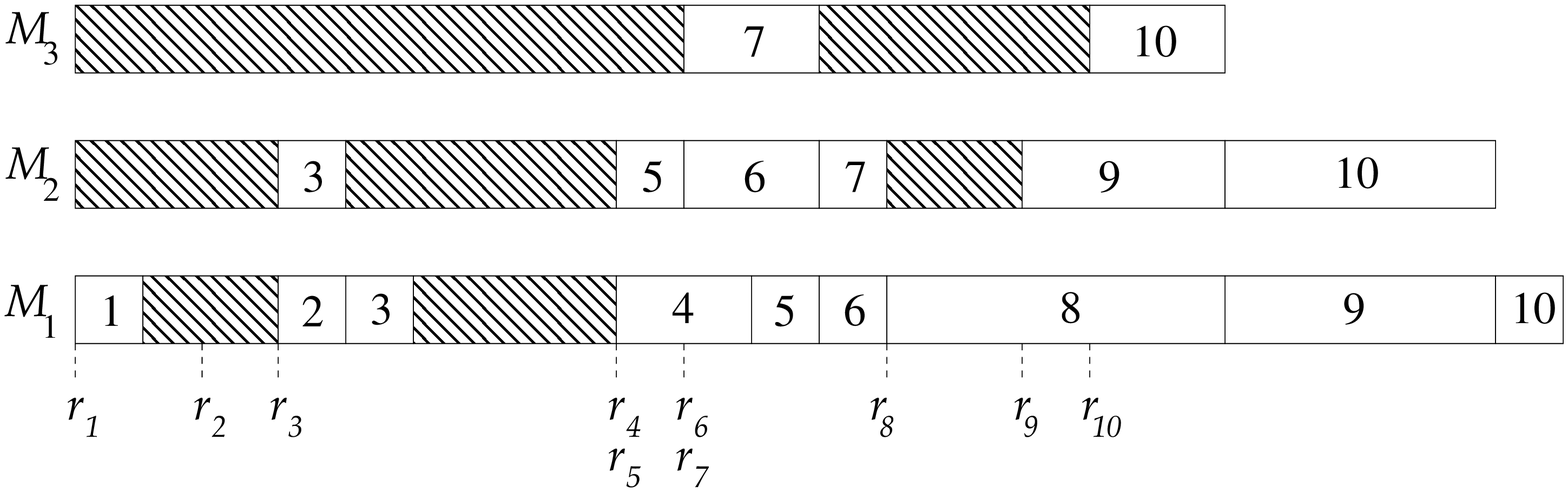}
\caption{A \FS\ schedule for \pbGenRelDates }
\label{fig:PFSGeneral}
\end{center}
\end{figure}

It is interesting to mention that for the problem $P|pmtn,r_j,p_j=p|\sum C_j$,
\cite{BBCDKS07} show that a similar structure is dominant.
This result is very specific since it deals with equal processing times and the
considered objective function is the total completion time.
In this paper, we prove the existence of such a structure for more general problems. 

Now, we  express our central result:

\begin{Theorem}\label{Theorem:PFS}
If \pbGenRelDates\ has a solution $S$ with completion times $C_1 \leq C_2 \leq \dots \leq C_n$, there exists a non-delay \FS\ solution $S'$ such that $C_1'\leq C_2' \leq \dots \leq C_n'$ and  $C_j' \leq C_j$ for $1 \leq j \leq n$.
\end{Theorem}

\begin{proof}
By Remark \ref{Rem:NonDelay-VerticalOrder}, we assume without loss of generality that $S$ is a non-delay schedule which is vertically ordered.
Let us define the following two properties:

\begin{description}
\item[$A(j)$:] If $i$ and $i'$ are two jobs such that $1 \leq i \leq j $ and $i<i'$ no machine processes a piece of job $i'$ before a piece of job $i$.
\item[$B(j)$:] “No machine processes more than one piece of job $i$, for $1 \leq i \leq j$.
\end{description}

If $A(n)$ and $B(n)$ are true, $S$ is a non-delay \FS\ solution. 
Otherwise, we prove by induction on the job number that $S$ can be transformed into a non-delay and vertically ordered schedule $S'$ such that $A(n)$ and $B(n)$ are true.

\begin{figure}[htbp]
\begin{center}
\includegraphics[scale=0.65]{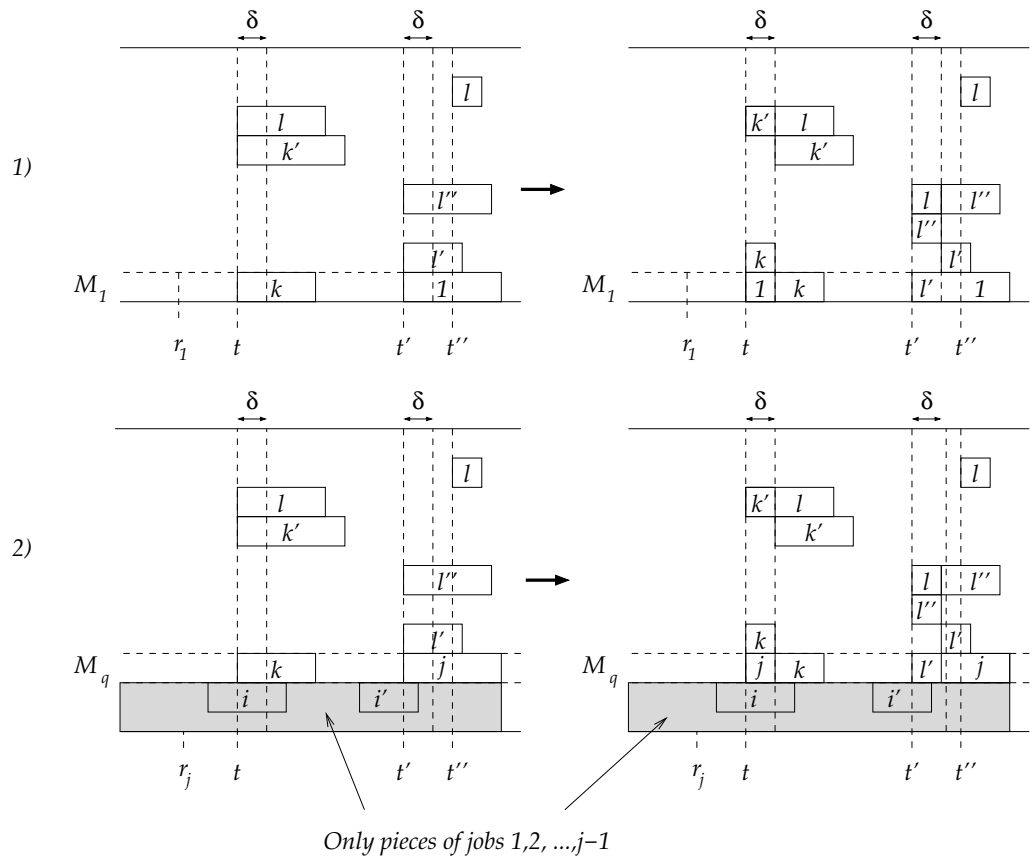}
\caption{The basic transformation used in Theorem \ref{Theorem:PFS}.}
\label{fig:PFSproof}
\end{center}
\end{figure}
%%%%%% Base induction
{\bf Base step.}

By Remark~\ref{Rem:NonDelay-VerticalOrder}, vertically-ordered schedules are dominant and hence 
each piece of job $1$ is processed by machine $M_1$. 

If $A(1)$ is true then $B(1)$ must be true: otherwise $M_1$ processes two pieces of job $1$ that are separated by an idle time interval, which contradicts the non-delay assumption.

So, suppose $A(1)$ is false and consider the smallest $t$ such that a job $k>1$ starts on $M_1$ at $t$ before a piece of job $1$. 
Let $t'$ be the starting time of this piece of job $1$ (see case~1 of Figure \ref{fig:PFSproof}). The non-delay property implies that all machines are busy during time interval $[t,t'[$.  
Since $|J(t)| = m$, and $|J(t')- \{1\}| \leq m-1$, there is a piece of a job $l$ processed at time $t$ that is not processed at time $t'$. Moreover, since $C_1 \leq C_l$ (by definition of $S$), there is also another piece of job $l$ that starts at time $t''>t'$.
Let $\delta = \min ( t''-t', \min_{i \in J(t) } \{ C(i, t) -t \}, \min_{i \in J(t')} \{ C(i, t') -t' \} )$. We exchange the piece of  job $l$ processed during time interval $[t, t+\delta[$ with the piece of job $1$ processed during time interval $[t', t'+\delta[$. Finally, we reassign pieces of jobs processed during time interval $[t, t+\delta[$ and pieces of jobs processed during time interval $[t', t'+\delta[$ respectively, so that the schedule remains vertically ordered. 
Only the completion time of job $1$ may decrease and the new schedule is still a non-delay one. 
Moreover, there is now a piece of job $1$ processed on machine $M_1$ during time interval $[t, t+\delta[$.
If $t>r_1$, the non-delay property of $S$ and the definition of $t$ imply that there is a piece of job $1$ processed during time interval $[r_1,t[$, so there is now one single piece of job $1$ processed during time interval $[r_1, t+\delta[$ 
{by machine $M_1$. If $A(1)$ is true, so is $B(1)$. Otherwise, we consider again the smallest $\tau$ such that a job $k>1$ starts on $M_1$ at $\tau$ before a piece of job $1$: since $\tau \ge t + \delta$ and $\delta > 0$ we have $\tau > t$. Therefore, by repeating this procedure, either we do not find such a $\tau$ or we reach the end of the schedule: in both cases $A(1)$ becomes true, so does $B(1)$.}
 
%%%%%% Etape induction
{\bf Induction step.}

Now, suppose there is a non-delay and vertically ordered schedule such that $A(j-1)$ and $B(j-1)$ are true for $j \ge 2$, and assume that at least one of the properties $A(j)$ and $B(j)$ is false.

%First, we show that there is a machine $M_q$ which processes a piece of a job $k>j$ starting at a time $t$, and a piece of job $j$ starting at time $t'>t$ (see case 2 of Figure \ref{fig:PFSproof}). 
%If $B(j)$ is true, it is obvious (otherwise $A(j)$ is true). Else, if $B(j)$ is false, there is one machine $M_q$ which processes two pieces of job $j$. Let $t$ and $t'>t$ be their completion and starting times, respectively. If the machine is idle during time interval $[t, t'[$, we get a contradiction. Indeed, if $q=1$ the non-delay assumption is not verified. If $q > 1$, let $i$ be a job processed at time $t'$ on a machine $M_p$, with $p < q$: since the schedule is vertically ordered, we have $i<j$, and since $A(j-1)$ is true, no piece of job $j$ can be processed by machine $M_p$ during time interval $[0, t'[$. Hence, no piece of job $j$ is processed by machines $M_1, M_2, \ldots, M_{q-1}$ during time interval $[t, t'[$, which contradicts the non-delay assumption. 

If $A(j)$ is true then $B(j)$ must be true. Indeed, suppose for the sake of contradiction that $B(j)$ is false. In this case, there is one machine $M_q$ which processes two pieces of job $j$. Let $t$ and $t'>t$ be their completion and starting times, respectively. If the machine is idle during time interval $[t, t'[$, we get a contradiction. Indeed, if $q=1$ the non-delay assumption is not verified. If $q > 1$, let $i$ be a job processed at time $t'$ on a machine $M_p$, with $p < q$: since the schedule is vertically ordered, we have $i<j$, and since $A(j-1)$ is true, no piece of job $j$ can be processed by machine $M_p$ during time interval $[r_1, t'[$. Hence, no piece of job $j$ is processed by machines $M_1, M_2, \ldots, M_{q-1}$ during time interval $[t, t'[$, and the piece of job $j$ starting at time $t'$ can start earlier at time $t$, which contradicts the non-delay assumption. 

So, suppose $A(j)$ is false, and let us consider the smallest value $t$ such that a piece of a job $k>j$ starts at time $t$, and a piece of job $j$ starts at time $t'>t$, on a machine $M_q$ (see case~2 of Figure \ref{fig:PFSproof}).
Since $k>j$ and the schedule is vertically ordered, no piece of job $j$ can be processed by machines $M_{q+1}, M_{q+2}, \ldots, M_m$ during time interval $[t, C(k,t)[$.
If $q=1$, no piece of job $j$ can be processed during this time interval. Else, let $i$ and $i'$ be the jobs processed respectively at time $t$ and $t'$ by machine $M_{q-1}$.
Since the schedule is vertically ordered, we have $i' < j$, and since $A(j-1)$ is true we have $i \leq i'$, so we get $i < j$.
Consequently, no piece of job $j$ can be processed by machines $M_1, M_2, \ldots, M_{q-1}$ during time interval $[t, C(k,t)[$. 

Since it is a non-delay schedule and $r_j \leq r_k \leq t$, all machines are busy during time interval $[t,t'[$: therefore there is a piece of a job $l>k$ processed at time $t$ that is not processed at time $t'$. Since $C_j \leq C_l$ (by definition of $S$), there is also another piece of job $l$ that starts at time $t''>t'$.
Let $\delta = \min ( t''-t', \min_{i \in J(t) } \{ C(i, t) -t \}, \min_{i \in J(t')} \{ C(i, t') -t' \} )$. We exchange the piece of job $l$ processed during time interval $[t, t+\delta[$ with the piece of job $j$ processed during time interval $[t', t' + \delta[$.
Finally, we reassign pieces of jobs processed during time interval $[t, t+\delta[$ and pieces of jobs processed during time interval $[t', t'+\delta[$ respectively, so that the schedule remains vertically ordered. 

Again, only the completion time of job $j$ may decrease, and we still have a non-delay schedule. Moreover there is a piece of job $j$ processed on machine $M_q$ during time interval $[t, t+\delta[$. If there is also a piece of job $j$ processed by $M_q$ before time $t$, it must end at time $t$ (by definition of $t$ and because of the non-delay assumption).
Hence, machine $M_q$ processes no piece of a job $k>j$ and only one single piece of job $j$ during time interval $[0, t+\delta[$. 
{Since $\delta > 0$, we can repeat the procedure, as in the "base step", and get a schedule with $A(j)$ and $B(j)$ being true.}
%Also, by reassigning the pieces of jobs processed during time interval $[t, t+\delta[$ and the pieces of jobs processed during time interval $[t', t'+\delta[$ respectively, property (a) remains valid. Moreover, since no piece of a job with a number greater than $i$ is processed by $M_q$ before time $t$, property (b) is true for job $i$ during time interval $[0, t+\delta[$. 
\end{proof}

\begin{Remark}\label{Remark:WithoutRelease}
For problems without release date, i.e. of the form \pbGen, non-delay \FS\ schedules are dominant since, by renumbering the jobs, there always exists an optimal solution such that $C_1^* \leq C_2^* \leq \dots \leq C_n^*$.
\end{Remark}

This theorem gives a very precise structure on an optimal solution for problems of the form \pbGenRelDates\
for which there exists an optimal solution $S^*$ such that $r_j < r_k \implies C_j^* \leq C_k^*$. Moreover, this structure is very interesting combined with a dominant order for the jobs' completion times, because it can lead to the
time-polynomiality of a large class of problems.
Indeed, the linear programming approach proposed by~\cite{BBCDKS07} for problem $P|pmtn, r_j, p_j=p|\sum C_j$ can be extended to regular criteria which are separable piecewise continuous linear, such as $\sum w_jT_j$ for example.

Let us modify their linear program. The value $p$ is replaced with $p_j$ in the job processing time constraints, and the objective function is replaced with a regular criterion $f$ which is a separable piecewise continuous linear function. Properties of $f$ ensure that it can be handled by a linear program (see for instance \cite{DT97}). We then get:

\begin{alignat}{4}
\nonumber \mbox{minimize } & f(C_1, C_2, \ldots,C_n) & &\\
\nonumber \mbox{subject to } & & &\\
\sum_{l=1}^{m} p_j^l & =  p_j&\;\;\;\; &\forall j = 1, \ldots,n \label{eqn:proc}\\
t_j^{l+1} + p_j^{l+1} & \leq t_j^{l}&\;\;\;\; &\forall j = 1, \ldots,n , \forall l = 1, \ldots,m-1\label{eqn:flowshop}\\
t_j^l + p_j^l & \leq t_{j+1}^l&\;\;\;\; &\forall j = 1, \ldots,n-1 , \forall l = 1, \ldots,m-1\label{eqn:prec}\\
t_j^m & \geq r_j&\;\;\;\; &\forall j = 1, \ldots,n \label{eqn:ri}\\
t_j^1 + p_j^1 & = C_j&\;\;\;\; &\forall j = 1, \ldots,n \label{eqn:Ci}\\
%T_j & \geq t_j^{m} + p_j^{m} - d_i&\;\;\;\; &\forall j = 1, \ldots,n \label{eqn:Ti}\\
p_j^l, t_j^l & \geq 0&\;\;\;\; &\forall j = 1, \ldots,n, \forall l = 1, \ldots,m \label{eqn:0}\\
C_j & \geq 0&\;\;\;\; &\forall j = 1, \ldots,n \label{eqn:0-bis}
\end{alignat}

The ${\cal O}(nm)$ variables $t_j^l, p_j^l$, and $C_j$ are defined respectively as the starting time of job $j$ on machine $M_l$, the processing time of job $j$ on machine $M_l$, and the completion time of job $j$.
Equalities~(\ref{eqn:proc}) guarantee the processing time of each job, whereas inequalities~(\ref{eqn:flowshop}) 
%are due to Theorem~\ref{Theorem:PFS} and 
ensure that a job is first processed on the different machines in the order $M_m$, $M_{m-1}$, $\dots$, $M_{1}$.
Inequalities~(\ref{eqn:prec}) ensure that jobs are scheduled on each machine following order $1,2, \ldots, n$.
 
\begin{Theorem}\label{Theorem:PL}
The problem \pbGenRelDates\ can be solved in polynomial time if $f$ is a separable piecewise continuous linear function computable in polynomial time and if there exists an optimal solution such that $C_1^* \leq C_2^* \leq \dots \leq C_n^*$.
\end{Theorem}

\begin{proof}
We use Theorem \ref{Theorem:PFS} with the idea of the proof in ~\cite{BBCDKS07}. Theorem \ref{Theorem:PFS} implies that there exists an optimal non-delay \FS\ solution such that the order of the jobs on each machine is $1,2, \ldots, n$. If we denote by $z^*$ its value, and by $w^*$ the minimum value of a \FS\ solution which verifies order $1,2, \ldots, n$ for the jobs, we have $z^* \leq w^*$.
Now, observe that any solution of the linear program defines a \FS\ schedule. Hence, if $f^*$ is the value of an optimal solution of the linear program, we get $f^* \ge w^*$, that is  $f^* \ge z^*$.
However, in a \FS\ schedule a job $j$ may have a completion time less than $C_j$ because there may be no piece of $j$ on machine $M_1$. So, let us show that the optimal non-delay \FS\ solution (denoted by $\sigma$) verifies the constraints of the linear program.  Suppose there is a job $j$ in $\sigma$  that completes at time $C^{\sigma}_j$ and whose last piece is processed by machines $M_q$ with $q > 1$. 
By Theorem \ref{Theorem:PFS} we know that there is no job $i <j$ such that $C^{\sigma}_i > C^{\sigma}_j$, and also that $\sigma$ is vertically ordered, so no job $i > j$ is scheduled by processors $M_1, M_2, \ldots, M_q$ before time $C^{\sigma}_j$: therefore, we can define a solution of the linear program such that $p^l_j = 0$ for $1 \leq l < q$, that is such that $C_j = C^{\sigma}_j$. By applying this procedure to any job $j$ whose last piece is not processed by machine $M_1$, we get a solution of the linear program of value $z^*$, which implies $z^* \ge f^*$. From $z^* \leq f^*$ we deduce that $f^* = z^*$: the optimal solution of the linear program is also an optimal solution of the problem \pbGenRelDates\ .
\end{proof}

Note that a more general LP formulation is proposed in~\cite{KW12}, since it is dedicated to solve the problem
$Q|r_j,pmtn,D_j,C_1\leq \dots \leq C_n|F$, where $C_1\leq \dots \leq C_n$ means that we are only
looking for an optimal schedule among the class of schedules for which $C_1\leq \dots \leq C_n$ holds.
Nevertheless, notice that the formulation introduced here is interesting for two reasons: first, 
this LP formulation involves only $O(nm)$ variables and $O(nm)$ constraints, whereas the one proposed
in~\cite{KW12} uses $O(n^3m)$ variables and constraints. Secondly, the approach provided in~\cite{KW12}
does not allow us to use the \FS\ structure proposed in this paper.

The next section is dedicated to finding subproblems of \pbGenRelDates\ for which a total order on jobs' completion times can be obtained, in order to conclude that they are solvable in polynomial time.

\section{Ordering the jobs' completion times}\label{sec:order}

Extending notations on the agreeability introduced in~\cite{TNC09},
in the $\beta$-field, we write $(r_j^+, p_j^+)$ if $r_j$'s and $p_j$'s are
in the same order, i.e. $p_1\leq p_2 \leq \dots \leq p_n$.
Note that, for a problem with equal processing times, or without release dates,
this condition is always fulfilled.
This notation can also be used for more than two inputs; for example, $(r_j^+, p_j^+, d_j^+, w_j^-)$
means that $r_j$'s, $p_j$'s and $d_j$'s are in an increasing order whereas $w_j$'s are decreasing.

Under some specific conditions on the objective functions and the input data, it is possible to know the order of the jobs' completion times in an optimal solution:

\begin{Theorem}\label{Theorem:agreeable}
The following problems are solvable in polynomial time:
\begin{enumerate}
\item $P|pmtn, (r_j^+, p_j^+) | \sum f_j$, when $f_j$'s are regular functions and $f_j-f_k$ is non-decreasing if $j<k$.
\item $P|pmtn, (r_j^+, p_j^+) | \max f_j$, when $f_j$'s are regular functions and $f_j-f_k$ is non-negative if $j<k$.
\item $P|pmtn, (r_j^+, p_j^+, w_j^-), d_j=d | \sum  w_j U_j$.
\end{enumerate}

%\begin{enumerate}
%\item $\delta = (r_j^+, p_j^+)$ and $f = \sum_j f_j$, where $f_j$'s are regular functions and $f_j-f_k$ is non-decreasing if $j<k$.
%\item $\delta = (r_j^+, p_j^+)$ and $f = \max _j f_j$, where $f_j$'s are regular functions and $f_j-f_k$ is is non-negative if $j<k$.
%\item $\delta = (r_j^+, p_j^+, d_j^-, w_j^-)$ and $f = \sum_j  w_j U_j$.
%\end{enumerate}

\end{Theorem}

\begin{proof}
We first show that there exists an optimal schedule such that $C_1^* \leq C_2^* \leq \dots \leq C_n^*$.
Let $S$ be an optimal schedule. For the sake of contradiction, assume that there exist two jobs $j<k$,
such that $r_j \leq r_k$, $p_j \leq p_k$ and $C_j > C_k$, and let us prove that we can find another optimal schedule $S'$ such that $C_j' \leq C_k'$.
This exchange argument is illustrated with Figure~\ref{fig:exchange}.
\begin{figure}
\begin{center}
\includegraphics[scale=0.4]{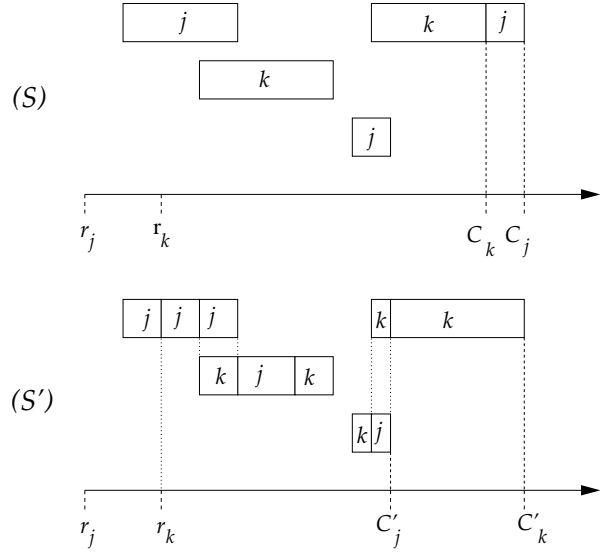}
\caption[]{An optimal solution and the corresponding optimal solution after the exchange.}
\label{fig:exchange}
\end{center}
\end{figure}

The schedule $S'$ is constructed in the following manner: all the pieces of jobs but $j$ and $k$ remain exactly at
the same position.
All pieces of job $j$ processed before $r_k$ stay at the same place, and on any time interval where jobs $j$ and $k$ are both processed in $S$, we schedule them in the same manner in $S'$. For the remaining available slots, starting from time $r_k$, we schedule the remaining part of job $j$ and then the one of $k$.
By construction, we ensure that no overlap exists in $S'$  by scheduling jobs $j$ and $k$ in $S'$ when $j$ and $k$ are simultaneously processed in $S$.
Using the fact that $p_j \leq p_k$, the completions times of $j$ and $k$ are $C_j' \leq C_k$ and $C_k' = C_j$. 

Now, let us consider the three cases:

1. Since $f_j - f_k$ is a non-decreasing function, by considering time points $C_j$ and $C_k$, we can write $(f_j - f_k)(C_j) \geq (f_j-f_k)(C_k)$, which means that $f_j(C_k)+f_k(C_j) \leq f_j(C_j)+f_k(C_k)$.

2. We have $\max (f_j(C_j'),f_k(C_k')) \leq \max (f_j(C_j), f_k(C_j))$ and, using the fact that $f_j-f_k$ is a non-negative function, we can write $\max (f_j(C_j), f_k(C_j)) = f_j(C_j) \leq \max (f_j(C_j),f_k(C_k))$.

%3. If $C_k \ge d$ or $C_j \leq d$, we have $w_j U'_j+w_k U'_k = w_j U_k + w_k U_j$. Else, $U'_j = U_k =% 0$ and $U'_k = U_j = 1$, so $w_j U'_j+w_k U'_k \leq w_j U_k + w_k U_j$.

3. We have to consider three cases according to the value of $d$; if $d \geq C_j$, we have $U'_j=U_j=0$ and $U'_k=U_k=0$. 
If $C_k > d$ then $U'_j \leq U_j=1$ and $U'_k=U_k=1$. In both cases we get $w_j U'_j+w_k U'_k \leq w_j U_j + w_k U_k$
Finally, if $C_j > d \geq C_k$ then $U'_j=U_k=0$ and $U'_k = U_j=1$, so we get $w_j U'_j+w_k U'_k  = w_k \leq w_j = w_j U_j + w_k U_k$ .
% 3. Since $C_j' \leq C_k$ and $C_k' = C_j$ we have $U'_j \leq U_k$ and $U'_k = U_j$, which implies $w_j U'_j+w_k U'_k \leq w_j U_k + w_k U_j$. And, because $U_k \leq U_j$ and $w_k \leq w_j$, we have $w_j U_k + w_k U_j \leq w_j U_j + w_k U_k$, that is $w_j U'_j+w_k U'_k \leq w_j U_j + w_k U_k$. 

Hence, $S'$ is also optimal in all cases, and there exists an optimal schedule such that $C_1^* \leq C_2^* \leq \dots \leq C_n^*$.
Therefore, problems of cases~1 and~2 are solvable in polynomial time by Theorem \ref{Theorem:PL}. 
For case~3, Theorem~\ref{Theorem:PL} cannot be used because $f_j$'s functions are not continuous. However, we only need to find the first job $k$ that cannot be on time, since jobs $l>k$ will be also late, by Theorem \ref{Theorem:PFS}. This can be done by testing the feasibility of successive modified versions of the linear program of Theorem \ref{Theorem:PL}: for $k= 1, 2, \ldots, n$, only on-time jobs $1 \leq j \leq k$ are considered, and constraints $t^1_j +p^1_j \leq d$, for $1 \leq j \leq k$, are added.  
\end{proof}

These results are used in the next section to derive time-polynomially solvable problems
according to different classical criteria.

\section{Consequences on classical criteria problems}\label{sec:discussion}

Figure \ref{fig:classificationTj} %and \ref{fig:classificationUj} 
shows the new complexity hierarchy of
parallel identical machine scheduling problems with preemption and release dates 
for criteria based on the $\sum w_j T_j$ . Previously minimal \NP-hard and maximal time-polynomially solvable problems are included with their references, and shaded results are the ones proved with 
the unified approach proposed in this paper. 

\begin{figure}[htbp]
	\begin{center}
%	\scalebox{1}{\input{scheme.pstex_t}}
\includegraphics{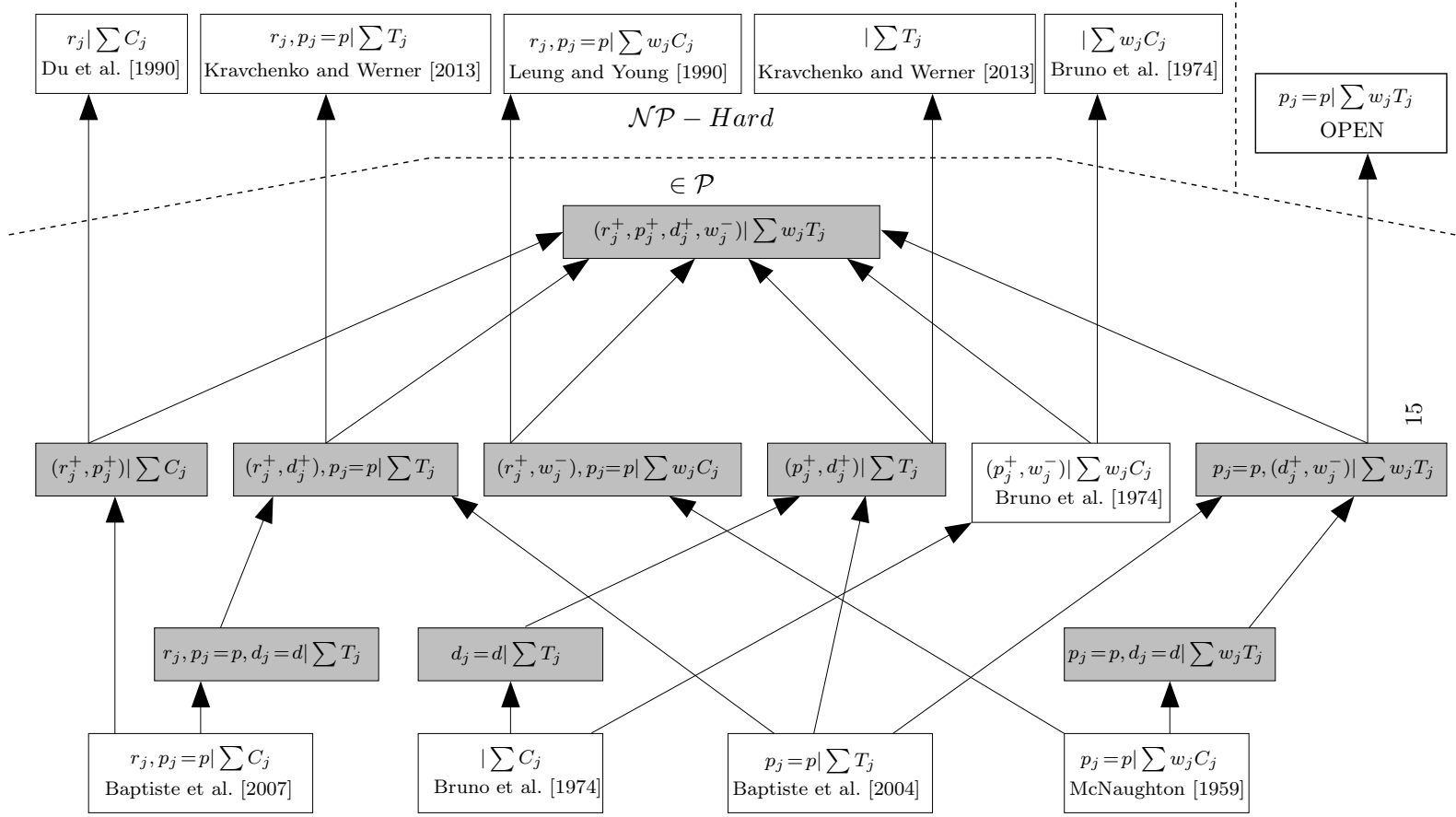}
	\caption{Classification of scheduling problems of type $P|pmtn,r_j|\sum w_j T_j$. ``$P|pmtn,$'' is omitted. An arc means ``is a subproblem of''. In gray, main results provided in this paper with a unified approach. }
	\label{fig:classificationTj}
	\end{center}
\end{figure}

%\begin{figure}[htbp]
%	\begin{center}
%	  \includegraphics[scale=1]{schemeUj2.eps}
%	\caption{Classification of scheduling problems of type $P|pmtn,r_j|\sum w_j U_j$. ``$P|pmtn,$'' is omitted. An arc means ``is a subproblem of''. In gray, new results provided in this paper.}
%	\label{fig:classificationUj}
%	\end{center}
%\end{figure}

For the four types of input data $r_j, p_j, d_j$ and $w_j$ there exist many agreeable combinations. In order to present a synthetic view, we consider cases where some of the data are constant, i.e. $r_j=0$, $p_j=p$, $d_j\in\{0,d\}$ or $w_j = 1$. 
Thus, we distinguish the two following cases:

a) Three types of data have constant values: all the problems are solvable in polynomial time, by Theorem~\ref{Theorem:agreeable}.
More precisely, if $d_j = 0$ or if $d_j$ is not fixed, these results were already known. It is interesting to notice that some of them were proved more than 30 years ago, whereas  the others are very recent (\cite{BBKT04} and \cite{BBCDKS07}). 
We extend these results and prove the time-polynomiality of the six problems we get by the combinations of parameters and criteria when $dj = d$, which is known in the literature as the 'common due date' case.

b) Two types of data have constant values: if the remaining data verify the condition of Theorem~\ref{Theorem:agreeable}, then the corresponding problem is in \P\; we hence provide a polynomial-time algorithm for different problems with agreeable conditions, for which very few results were known. 
Otherwise, i.e. if the remaining data do not necessary satisfy the condition of Theorem~\ref{Theorem:agreeable}, using a literature review, we observe that the problem is either \NP-hard or open. For flow-time or tardiness criteria, all problems are \NP-hard, except $P|pmtn,p_j=p|\sum w_jT_j$: that is why we conjecture it is \NP-Hard.

Our approach also defines new problems with common due dates for which complexity issues are interesting for criteria related to weighted total number of late jobs.
When we have to deal with common due dates and a criterion based on $\sum w_jU_j$, it is possible
to look at the reverse problem and hence provide directly complexity results.
For example, $P|r_j,d_j=d, pmtn|\sum U_j$ is equivalent by symetry to $P|pmtn|\sum U_j$ and is hence
NP-Hard. In the same way, problem $P|pmtn, r_j, p_j=p,d_j=d|\sum U_j$ is equivalent to $P|pmtn, p_j=p|\sum U_j$, for which a polynomial time algorithm was given in~\cite{BBKT04}.
Our approach leads to a slightly more general result since we show that $P|pmtn,(r_j^+,p_j^+,w_j^-),d_j=d|\sum w_jU_j$ is solvable in polynomial time.

One can also note that our approach does not perform well for a criterion of type $f = \max f_j$. Indeed, we can only prove that $P|pmtn|C_{max}$ and $P|pmtn, p_j=p|L_{max}$ are in \P, whereas it was independently proved in~\cite{LL78} and~\cite{S81} that $R|pmtn, r_j|L_{max}$ can be solved in polynomial time.

\section{Conclusion}\label{sec:conclu}

In this paper we proved that there exists a type of schedules (named \FS) which is dominant for problem \pbGenRelDates\ 
if there exists an optimal solution with completion times scheduled in the same order as the release dates, 
or if there is no release date.
By interchange arguments, we proved that, for a large subclass of these problems, it is possible to order the optimal 
completion time of all jobs. Using these two results, we showed that problems satisfying the condition
of Theorem~\ref{Theorem:agreeable} are polynomially solvable. In particular, we proved the polynomiality of different
problems having agreeable data and/or a common due date $d_j=d$. 
\begin{figure}[htbp]
\begin{center}
\includegraphics[scale=0.4]{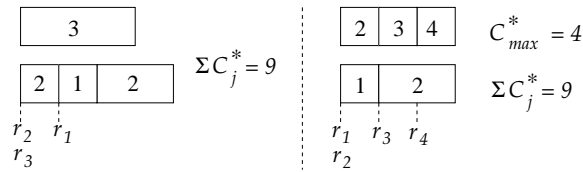}
\caption{Finding a better characterization of \FS\ schedules}
\label{fig:counter}
\end{center}
\end{figure}

An interesting question consists in finding a less restrictive condition for the existence of \FS\ schedules.
Indeed, on the one hand, \FS\ schedules are not dominant if condition $(r_j^+, p_j^+)$ does not hold (see left chart of Figure~\ref{fig:counter}) but, on the other hand, there exist \FS\ schedules that are optimal, even if the condition $(r_j^+, p_j^+)$ is not verified (right chart of Figure~\ref{fig:counter}). Hence, we should seek for a better characterization of the conditions implying the existence of  \FS\ schedules.

Another research avenue lies in using the \FS\ structure in \NP-Hard problems to derive approximation algorithms or to improve resolution methods; indeed, the existence of \FS\ schedules drastically reduces the combinatoric of problems since, whatever the number of machines, there are at most $n!$ orders to test.

Finally, polynomial-time cases for problem \pbGenRelDates\ that were already solved in the literature
are also polynomially solvable when we generalize to the uniform machine case $Q|pmtn,r_j|f$, leading to 
the natural following question: is the \FS\ structure also dominant for some problems in case of uniform machines?

\nocite{K00}
\nocite{LY90}
\nocite{M59}

\section*{Acknowledgements}
We would like to thank the anonymous referees for their
valuable suggestions and constructive comments that improved
this paper.

\bibliographystyle{elsarticle-harv}
\bibliography{Biblio}
\end{document}